%% file: main.tex
\documentclass[10pt,conference]{IEEEtran} 
\usepackage{enumitem}

\input{wok-commands.tex}

\usepackage{subcaption}

\usepackage[normalem]{ulem} 

\usepackage{xcolor}
\newcommand{\bing}[1]{{\color{blue}BW: #1}}

\begin{document}

\title{S-CAD: Selective Classical Advantage Distillation for Quantum Conference Key Agreement}

\author{\IEEEauthorblockN{Trevor Thomas}
\IEEEauthorblockA{\textit{School of Computing} \\
\textit{University of Connecticut}\\
Storrs CT, USA \\
trevor.thomas@uconn.edu}
\and
\IEEEauthorblockN{Walter O. Krawec}
\IEEEauthorblockA{\textit{School of Computing} \\
\textit{University of Connecticut}\\
Storrs CT, USA \\
walter.krawec@uconn.edu}
\and
\IEEEauthorblockN{Bing Wang}
\IEEEauthorblockA{\textit{School of Computing} \\
\textit{University of Connecticut}\\
Storrs CT, USA \\
bing@uconn.edu}
}

\maketitle
  
\begin{abstract}
Quantum conference key agreement (QCKA) protocols utilize GHZ states to establish shared group keys between multiple parties.  While previous work has shown that standard Classical Advantage Distillation (CAD) protocols can sometimes benefit QCKA performance, it was unknown if past results were asymptotically tight.  In this work, we design a new CAD protocol, {\em Selective  Classical Advantage Distillation (S-CAD)}, for QCKA,
which generalizes prior QCKA+CAD work and allows the parties to selectively enable or disable CAD.  We derive an asymptotic proof of security against general coherent attacks, which outperforms prior work.  Finally, we evaluate in a variety of simulated star network topologies, showing when S-CAD can help, and when it is best to disable CAD entirely.
\end{abstract}


\input{intro}

\input{prelim}
\input{protocol}

\input{asym-proof}

\input{eval}

\input{closing}
\balance
\bibliographystyle{unsrt}
\bibliography{paper/bib/groupCAD,conference/local-bib}

\end{document}

%% file: wok-commands.tex
\usepackage{amsthm,amsmath,amsfonts,braket,algorithm,graphicx,braket,balance,mathtools}

\theoremstyle{definition} 
\newtheorem {theorem} {Theorem}

\newtheorem {lemma} {Lemma}

\usepackage{ifthen}
\newboolean{fullversionflag}
\setboolean{fullversionflag}{true}
\newcommand{\fullversion}[2]{\ifthenelse{\boolean{fullversionflag}}{{#1}}{{#2}}}

\newcommand{\heading}[1]{\text{ }\newline\textbf{{#1:}}}

\newcommand{\num}{\#}

\newcommand{\bk}[1]{\braket{#1|#1}}
\newcommand{\kb}[1]{\left[#1\right]}

\newcommand{\st}{\text{ } : \text{ }}
\newcommand{\stt}{\text{ s.t. }}

\newcommand{\leakEC}{\texttt{leak}_{EC}}

\newcommand{\up}[1]{^{({#1})}}

\newcommand{\G}{{\mathcal{G}}}

\newcommand{\g}[2]{{ \left[\begin{smallmatrix} {#2}\\{#1} \end{smallmatrix} \right] }}

\newcommand{\B}{\mathcal{B}}
\newcommand{\CAD}{{\mathcal{C}}}
\newcommand{\cad}{{\mathcal{C}}}

\newcommand{\wforall}{{\text{ } \forall \text{ }}}

\floatname{algorithm}{Protocol}

\newcommand{\wok}[1]{{\color{blue}WOK: #1}}

%% file: intro.tex
\section{Introduction}
 Quantum conference key agreement (QCKA), the multi-user variant of quantum key distribution (QKD), aims  to establish a secret \emph{group} key, shared among a group of users.
Group keys are vital cryptographic resources for many distributed cryptographic applications such as  secure multiparty communication, distributed cryptographic applications, and quantum secure multi-party computation  (see \cite{murta2020quantum} for a review).  
While multiple, parallel, instances of 
pairwise QKD may be used to achieve conference key agreement,
QCKA protocols incur significantly less quantum  resources and lower classical communication overhead. 
In addition, there are other potential advantages to QCKA protocols such as higher key-rates in some scenarios \cite{epping2017multi}.
Existing studies (e.g., \cite{proietti2021experimental}) have also shown that QCKA protocols are experimentally feasible.

In general, QCKA, like QKD, is limited in distance, making quantum networks, capable of distributing GHZ states \cite{meignant2019distributing,fischer2021distributing,bugalho2023distributing,oslovich2025efficient}, a vital tool in their eventual large-scale adoption.
Many factors (e.g., loss and noise in quantum channels, short quantum memory decoherence time, and so on) can affect the quality of a quantum network, and hence any QCKA application supported by the network. 
Furthermore, it is likely that in a large network scenario, some users 
may experience significantly higher noise than others, leading to heterogeneous
conditions \cite{krawec2025quantum}. We therefore need to develop QCKA techniques that yield
efficient group key generation in both 
homogeneous settings (where the observed noise for all parties is equal) and 
heterogeneous settings (where some parties have a higher noise level than
others). The asymmetric setting is particularly challenging since a user with significantly inferior channel conditions than others can degrade the overall key rate for all the users.

In 
two-party 
QKD, it has been shown that \emph{classical advantage distillation} (CAD) \cite{maurer1993secret,bae2007key} can greatly improve the noise tolerance of the system.  CAD is a purely classical protocol that runs after the quantum communication stage is complete.  Using two-way classical communication, CAD will transform the given raw key (from the quantum stage of the QKD protocol) into a smaller version, but also one with less noise (and, thus, there will be less error correction leakage).  In high noise channels, this can greatly improve the efficiency of two-party QKD systems, while also increasing the overall noise tolerance significantly.

While CAD has been studied extensively in the two-party setting (see \cite{du2025advantage,sun2025enhancing,tan2020advantage,treplin2025finite,li2022improving} for some examples), in the group-key setting, however, it is less understood.  Applying CAD to a QCKA protocol is challenging for two reasons. First, it is not clear what the best CAD method is for multi-user settings. Indeed, the multi-party scenario presents many challenges and opportunities for protocol design. Second, a new proof of security is necessary, as the two-way classical communication to each party leaks additional information to an adversary.

Recently, in \cite{krawec2025quantum}, we developed a CAD protocol, inspired by the standard two-party version \cite{maurer1993secret,bae2007key}, for the multi-party scenario and applied it to the QCKA protocol introduced in \cite{grasselli2018finite}. That study was the first to  apply CAD techniques to this QCKA protocol.
In that work, we developed a finite key proof of security, and showed that CAD can potentially benefit the overall secret key generation rate and noise tolerances, depending on the number of parties and the network noise scenario.  Specifically, our earlier work showed a benefit when the number of parties was small and when the noise was heterogeneous.
For other settings, our earlier work did not outperform the standard ``No CAD'' case.  \emph{However, our earlier proof was not asymptotically tight}, as we demonstrate in this work, leaving it still an open question as to when and how CAD can be advantageously applied to group-key systems.  
 


In this work, we revisit the use of CAD for QCKA protocols.  Specifically, we design a new CAD protocol, \emph{Selective CAD (S-CAD)}, that takes into account multi-party scenarios and 
 allows some parties to selectively enable or disable CAD.  Our evaluations (\S\ref{sec:results}) show that this ability can greatly improve overall secret key generation rates for all parties.


We prove security of our protocol against general attacks in the asymptotic scenario and derive a general expression for the overall secret key rate.
While our proof is not the finite key scenario,
our methods can readily be used to derive a finite key bound using post selection techniques \cite{christandl2009postselection}, though we leave the exact details as future work.  Importantly, our asymptotic key-rate expression allows us to better understand the theoretical capabilities of CAD applied to QCKA protocols.



Finally, we perform an extensive evaluation of our S-CAD protocol in a variety of network scenarios.  We evaluate with a varying number of parties (ranging from three to eight), including both homogeneous and heterogeneous network scenarios.  
We show when S-CAD can benefit performance, and when it can hurt (and thus should be disabled).  We also show when it is important to disable CAD for some users, but not for  others, i.e., when to be selective.

Taken together, our work sheds new light on the importance of carefully designing classical post-processing techniques, specifically for quantum conference key agreement protocols.  It also sheds light on the importance of utilizing such protocols in networks where noise may not be homogeneous. 



%% file: prelim.tex
\subsection{Preliminaries}\label{section:prelim}

Let $\ket{\psi}$ be some (pure) quantum state.  We write $\kb{\psi}$ to mean $\ket{\psi}\bra{\psi}$.  We also define $P(\ket{\psi}) = \kb{\psi}$.  Given a density operator $\rho_{AB}$, we write $\rho_B$ to mean the result of tracing out the $A$ system; namely $\rho_B = tr_A\rho_{AB}$.  Similarly for three or more systems.

Given a classical random variable $X$, we write $H(X)$ to mean the Shannon entropy of $X$.  If $X$ takes values $x_i$ with probability $p_i$, then $H(X) = -\sum_ip_i\log_2p_i$.  We write $H(X|Y)$ to mean the conditional Shannon entropy, namely $H(X|Y) = H(XY) - H(Y)$.  Finally, we write $h(x)$, for $x\in[0,1]$, to mean the binary Shannon entropy function, namely $h(x) = -x\log_2x - (1-x)\log_2(1-x)$.

Given density operator $\rho_{AB}$, we write $H(AB)_\rho$ to mean the von Neumann entropy of $\rho$.  This is defined to be $H(AB)_\rho = -tr(\rho_{AB}\log_2\rho_{AB})$.  We write $H(A|B)_\rho$ to mean the conditional von Neumann entropy, namely $H(A|B)_\rho = H(AB)_\rho - H(B)_\rho$, where $H(B)_\rho$ is the von Neumann entropy of $\rho_B = tr_A\rho_{AB}$.  Note that, on classical systems, von Neumann entropy agrees with Shannon entropy.  If the context is clear, we may forgo writing the subscript.  Given a classical-quantum-classical state of the form $\rho_{AEZ} = \sum_zp_z\rho_{AE}\up{z}\otimes\kb{z}$, then it can be shown that:
\begin{equation}\label{eq:prelim:mixed}
H(A|EZ) \ge \sum_zp_zH(A|E,Z=z) = \sum_zp_zH(A|E)_{\rho\up{z}}.
\end{equation}

\heading{QCKA Security}
A QCKA protocol involves the establishment of a \emph{raw key} between Alice and each Bob (we assume there are $p \ge 2$ Bobs).  Let $A$ represent Alice's raw key register and $B\up{j}$ represent the $j$th Bob's raw key register.  Ideally, $A = B\up{1} = B\up{2} = \cdots =B\up{p}$; however, these raw keys are noisy, and an adversary potentially has quantum side information on them.  Thus, parties will need to run an error correction protocol (which leaks additional information) and a privacy amplification protocol.  The latter involves choosing a two-universal hash function and hashing the raw key to a smaller, but \emph{secret}, key of size $\ell$-bits.  For more information on the definition of security for a QKD protocol and the impact of error correction and privacy amplification, the reader is referred to \cite{renner2008security}.

In the asymptotic setting, one typically considers security against collective attacks  first, and then one may use post-selection techniques \cite{christandl2009postselection} to promote this to security against general attacks.  Let $\rho_{ABE}$ be the quantum state representing a single raw key bit for all parties (where $B$ can be decomposed into $B\up{1}\cdots B\up{p}$), while the $E$ system is arbitrary.  Then, if the raw key size is $n$-bits, the joint system, assuming collective attacks, is $\rho_{ABE}^{\otimes n}$.  Using results in \cite{epping2017multi,devetak2005distillation,renner2008security}, it can be shown that the asymptotic key-rate, after error correction and privacy amplification, is:
\begin{equation}\label{eq:prelim:keyrate}
\lim_{n\rightarrow \infty}\frac{\ell}{n} = H(A|E)_\rho - \max_jH(A|B_j).
\end{equation}
Thus, to prove security of a QCKA protocol, and to determine its efficiency, one requires good bounds on $H(A|E)_\rho$.

Later, when analyzing the security of our protocol, we will need to compute the entropy in a system $\rho_{AE}$ of the form:
\begin{equation}\label{eq:prelim:cq}
\rho_{AE} = \frac{1}{M}\kb{0}_A \otimes\left(\sum_i\kb{F\up{0}_i}\right) + \frac{1}{M}\kb{1}_A\otimes\left(\sum_i\kb{F\up{1}_i}\right),
\end{equation}
where the $\ket{F\up{a}_i}$ are arbitrary (potentially sub-normalized) pure states in Eve's ancilla and $M$ is a normalization term.  For such a state, the following theorem from \cite{krawec2017quantum} can be used to bound the von Neumann entropy:
\begin{theorem}\label{thm:prelim}
  (From \cite{krawec2017quantum}): Let $\rho_{AE}$ be written in the form of Equation \ref{eq:prelim:cq}.  Then it holds that:
  \begin{equation}\label{eq:hAEGeneral}
    H(A|E)_\rho \ge \sum_i\left(\frac{\braket{F\up{0}_i|F\up{0}_i} + \braket{F\up{1}_i|F\up{1}_i}}{M}\right) H_i,
  \end{equation}
  where:
  \begin{equation}
    H_i = h\left(\frac{\bk{F\up{0}_i}}{\bk{F\up{0}_i} + \bk{F\up{1}_i}}\right) - h\left(\nu_i\right)
  \end{equation}
  (note that if $\bk{F_i\up{a}} = 0$, then one sets $H_i = 0$)
  and, finally $\nu_i = $
  \begin{equation}
\frac{1}{2} + \frac{\sqrt{ \left(\bk{F\up{0}_i} - \bk{F\up{1}_i}\right)^2 + 4Re^2\braket{F\up{0}_i|F\up{1}_i}}}{2\left(\bk{F\up{0}_i} + \bk{F\up{1}_i}\right)}.
  \end{equation}
  
\end{theorem}

\heading{GHZ Basis}
We will utilize the GHZ basis extensively in this work, and so we take some time
to describe our notation system for this basis, which will be required to understand our technical proofs later.  A $p+1$ party GHZ state, consisting of $p+1$ qubits, is indexed by $x \in \{0,1\}^p$ and $y\in \{0,1\}$ and defined to be
\begin{equation}\label{eq:prelim:gxy}
  \ket{g(x;y)} = \frac{1}{\sqrt{2}}(\ket{0,x} + (-1)^y\ket{1,\bar{x}}),
\end{equation}
where $\bar{x}$ is the bit-wise complement of $x$ (i.e., $\bar{x} = x\oplus (1\cdots 1)$).
Clearly, this is an orthonormal basis of a $p+1$ qubit space.  Note that the ``$x$'' portion represents ``bit errors'' between parties (thus we will often call it the \emph{bit component}).  The $y$ portion represents a phase error. Its value can be determined by measuring all $p+1$ qubits of the GHZ state in the $X$ basis, and adding the result modulo two.
Thus we will often call the $y$ portion the \emph{phase component} of the GHZ state.

Sometimes (for instance, in Lemma \ref{lem:DelayedMeasure}), we may need to work with $n$ GHZ states (or a superposition of such states).  To help with indexing of such a state, we define the set $\B_p = \{0,1\}^p$, and then, given a ``word'' $x \in \B_p^n$, we mean $x = x_1x_2\cdots x_n$, with each $x_i$ a $p$-bit string, namely $x_i \in \B_p$.  If we ever require the $j$th
bit of the 
$i$th
``block'' of $p$-bits, we will write $x_{i,j}$.

With this notation defined, given $x \in \B_p^n$ and $y\in \{0,1\}^n$, we write $\ket{g(x;y)}$ to mean $\ket{g(x;y)} = \ket{g(x_1;y_1)}\ket{g(x_2;y_2)}\cdots \ket{g(x_n;y_n)}$.  In general, the size of $p$ will be fixed, and known, so there should be no confusion in how to decompose $x \in \B_p^n$.  Finally, we will define the \emph{GHZ alphabet} to be $\G_p = \{\g{x}{y} \st x\in\B_p, y\in\{0,1\} \}$.  
Specifically, a character of this alphabet is an element written $\g{x}{y}$, where the lower coordinate represents the bit-flip string of the GHZ state and the top coordinate represents the phase flip bit.  We can then define $\G_p^n$ to be the set of \emph{GHZ words}, such that $q \in \G_p^n$ can be written $q = \g{x}{y} = \g{x_1\cdots x_n}{y_1\cdots y_n}$, with $x \in \B_p^n$ (i.e., $x_i \in \{0,1\}^p$), and $y \in \{0,1\}^n$.  The single character $q_i$, then is the 
$i$th character, namely $\g{x_i}{y_i}$, which represents one of the $p+1$ qubit GHZ states, namely $\g{x_i}{y_i} = \ket{g(x_i;y_i)}$, where the latter is defined in Equation \ref{eq:prelim:gxy}.

For example, if $p=3$, then $\g{101}{1}$ represents the GHZ state $\ket{g(101;1)} = \frac{1}{\sqrt{2}}(\ket{0101} - \ket{1010})$.  While $\g{101,001}{1,0}$ represents \emph{two} GHZ states, namely $\ket{g(101; 1)}\otimes \ket{g(001;0)}$.

%% file: protocol.tex
\section{The Protocol}\label{sec:protocol}

In this section, we discuss how our S-CAD protocol operates and how it integrates into an existing QCKA protocol.  Below, we first discuss the QCKA protocol introduced by \cite{grasselli2018finite}.  Following this, we will discuss our new S-CAD protocol, which operates as a secondary stage to the QCKA protocol.  Note that our S-CAD protocol can be incorporated into other QCKA protocols, however one may need to derive a new security proof for any alternative combination.

\heading{QCKA Protocol}
We consider the QCKA protocol introduced in \cite{grasselli2018finite}.  The goal of a QCKA protocol is to establish a shared, group key, between $p+1$ parties.  We designate a ``leader'' party and call this party Alice; the remaining $p$ parties are typically called Bob$_1$ through Bob$_p$ (or ``the Bobs'').  This protocol involves the distribution of GHZ states, ideally of the form $\ket{g(0\cdots 0;0)}^{\otimes N}$, where $N$ is the number of rounds the protocol runs, and each GHZ state consists of $p+1$ qubits.  For each of these states, Alice holds one of the qubits, while the $p$ Bobs hold the remainder.
Thus, each party has $N$ qubits. 
Alice may create these states and distribute qubits to each Bob or, alternatively, a quantum network (potentially controlled by the adversary) may create these entangled states between parties.

For some of the GHZ states, Alice and the Bobs will measure in the $Z$ basis, record their outcomes, and the Bobs will send their results to Alice who computes the expected $Z$ basis error rate.  Normally, in the standard QCKA protocol \cite{grasselli2018finite}, one needs to estimate the values $Q_{AB_i}$, which represents the probability of a bit-flip error between the leader, Alice, and each Bob$_i$, separately (thus, one needs to estimate $p$ different $Z$ basis error rates).  To derive a tight asymptotic key-rate proof for S-CAD, we will actually require Alice to estimate $Q^Z_\Delta$ for all possible $\Delta\in\{0,1\}^p$, which will represent the probability that Alice and each Bob$_i$ has an error, for every $\Delta_i = 1$.  For instance, if $p=3$, then $Q^Z_{011}$ represents the probability that both Bob$_2$ and Bob$_3$ have an error (relative to Alice's result), while Bob$_1$ has the correct result (again, relative to the leader, Alice).

Besides the $Z$ basis error rate, we will also need to sample some of the distributed GHZ states in the $X$ basis.  For this, all parties measure in the $X$ basis and the Bobs send their outcomes to Alice who will sum these outcomes, modulo two.  Ideally, this sum should be zero (see Section \ref{section:prelim}) and any non-zero sum is counted as an error.  Specifically, Alice estimates $Q_X$, the probability of an $X$ basis error, that is, the probability that all $X$ basis measurements, when summed together, modulo two, adds to one.

The remaining 
system,
which we will assume to be $2n < N$ signals for simplicity, will be measured by all parties in the $Z$ basis.  This will constitute Alice and the Bobs' raw keys.  Normally, in QCKA, error correction and privacy amplification are now run.  However, before performing these steps, parties will run our S-CAD protocol.


\heading{S-CAD}
S-CAD takes inspiration from the two-party CAD protocol from \cite{bae2007key} and is an extended version of the CAD protocol used in \cite{krawec2025quantum}.  
It is specifically designed for multi-party scenarios and allows parties to selectively enable or disable CAD operations.
Let $\cad \in \{0,1\}^p$ be a flag register, such that if $\cad_i = 1$, then Bob$_i$ will have CAD ``ON'' and, otherwise CAD is turned ``OFF.''  We assume $\cad \ne 0\cdots 0$, since if that is the case, then one would simply run the standard QCKA protocol
and use the
results from \cite{grasselli2018finite} to determine the final secret key size.  Thus, to run S-CAD,
we assume at least one Bob has CAD turned ON.  Note that we also assume the setting of $\cad$ is public knowledge and, thus, Eve also knows its setting.

Now, first Alice and Bob randomly permute their remaining systems and select $n$ of them to be their ``Left'' half and $n$ of them to be their ``Right'' half.  For each $j = 1,2, \cdots, n$, Alice will compute the parity of the $j^{th}$ bit of her Left and Right raw key bit.  This parity is sent through the authenticated classical channel to each of the Bobs (this is not secret, thus this information leakage must be accounted for in the proof).

Now, if Bob$_i$ has CAD turned on (i.e., $\cad_i = 1$), then he will examine his $j^{th}$ Left and Right raw key bit.  If the parity of his outcomes match that of Alice, he will signal to Alice to ``Accept'' this $j^{th}$ round; otherwise, 
i.e., the parity does not match, 
he ``Rejects'' this round.  If $\cad_i=0$ (i.e., Bob$_i$ has CAD turned off), this Bob will always signal to ``Accept.''  If all Bobs signal to Accept, then the $j^{th}$ round is accepted.  If at least one Bob  signals to Reject, then that round is rejected by all parties (including those Bobs who have CAD turned OFF).

For all rounds that were accepted, all parties discard the Right portion of the raw key, and keep only the Left.  This will be their new raw key.  They will then run error correction and privacy amplification on this new, shorter, raw key. 

The reason that this may produce a better secret key rate (depending on the channel noise) is due to the fact that, conditioned on acceptance, the error rate should be lower in the new raw key for any Bob that had CAD on.  Thus, error correction will leak less information.

Note that the decision to use S-CAD or not and the exact setting for $\cad$ can be determined \emph{after} parties have measured all their GHZ states and have determined the error rates in the channel.  Indeed, one can compute the expected key-rate for all scenarios (using our main result, below, in Section \ref{sec:asymproof}), and determine which setting produces the highest key-rate.  Thus, the parties can always choose the optimal settings to maximize their key-rate.

\heading{Delayed Measurement S-CAD}
Of course, the quantum communication stage of the QCKA protocol, where GHZ states are distributed to Alice and the Bobs, can be equivalently represented by an entanglement-based protocol, where Eve, the adversary, will produce the initial GHZ states used by the protocol (and these will
be potentially entangled with her ancilla).

Later, in our proof, it will be beneficial to analyze a \emph{delayed-measurement} version of the S-CAD protocol, which is essentially a purification of the S-CAD protocol.  Here, instead of Alice and each Bob measuring their system (after sampling) in the $Z$ basis, followed by computing and subsequently
checking the parity of their bits, they will perform all operations coherently using CNOT operations, storing their parity checks in additional ancilla registers (these registers will model the message that is typically sent, classically).  Later, measuring these ancilla registers will produce an equivalent post-measured state, compared to the actual protocol.  In particular, we will create two additional registers: $M$ and $rej$, both initialized to $\ket{0\cdots 0}$.  The first will represent the parity message that Alice sends and it will consist of $n$ qubits (one for each bit of the classical message).  The second will represent the message to ``accept'' or ``reject'', sent from each Bob, thus it will consist of $np$ qubits.  Note that, even if CAD is switched off for a particular Bob, that Bob will still send the message ``accept'' in our protocol.

For each round $i = 1, \cdots, n$, Alice will first apply a Double-Control NOT (DCNOT) gate, controlled on her Left and Right qubits for that round, and targeting the $i$th empty ancilla in $M$.  A DCNOT gate is simply two applications of a CNOT gate.  Specifically, given basis state $\ket{x,y,z}$, where $x,y,z\in\{0,1\}$, then if the first two qubits are the two controls, and the third qubit is the target, it holds that $DCNOT\ket{x,y,z} = \ket{x,y,z\oplus x\oplus y}$.  This action models Alice's parity computation and message (since her $M$ register is initially $0$).  Later she can measure her $M$ register and her Left and Right qubits, and the resulting post measured state will be identical to the actual protocol.

Now, each Bob will perform a parity check by performing DCNOT operations on their qubits, and checking to see if it matches the sent message (in the $M$ ancilla).  A flag of $1$ will be placed in the corresponding $rej$ register if a particular Bob rejects.  It is not difficult to show that performing these quantum operations and then, at a later time, measuring all registers in the computational basis, is equivalent to measuring in the computational basis first, and then performing the XOR operations classically.  Later, in our proof, it will be convenient to analyze this delayed measurement version, keeping the system as a pure state for as long as possible.  We will analyze the behavior of this protocol in Section \ref{sec:technical:delayed-cad}.


%% file: asym-proof.tex
\section{Security Proof}
\label{sec:asymproof}

We now prove security of the S-CAD QCKA protocol.  We will actually analyze the delayed measurement version, as discussed in Section \ref{sec:protocol}.  Before stating our main result, we will first prove a 
lemma that shows how the delayed measurement version acts on GHZ basis states.  This will be useful later, when analyzing the actual protocol state.  We also derive the result on general GHZ states, and so this result may be useful for future multi-party cryptographic protocols and security proofs, which rely on multiple GHZ states.

\input{technical}

\subsection{Bounding Eve's Uncertainty}


We are now in a position to state and prove our main result, namely an asymptotic key-rate bound for the QCKA protocol with S-CAD.  We will first prove security against collective attacks, where Eve attacks each round of the protocol independently and identically.  Later, we will promote this to security against general attacks.  In particular, for every round of the protocol, Eve will create an arbitrary quantum state $\ket{\psi}_{ABE}$, where the $A$ register is a single qubit, and the $B$ portion consists of $p$ qubits (one for each Bob).  The Eve system is arbitrary.  We will assume Eve has a perfect quantum memory to store all her ancilla for all rounds of the protocol.

Security against collective attacks involve bounding the von Neumann entropy, as discussed in Section \ref{section:prelim}.  In particular, after the protocol runs, a raw key is established, where a single bit of the raw key can be modeled by the density operator $\rho_{ABEM}$, where $M$ represents all messages sent, publicly, over the authenticated channel (in particular, the parity announcements; the Accept/Reject signals are also sent, however if we condition on a bit being distilled, all Bob's signal to Accept, so there is no information to be gained in this case).  Note that this is the conditional state, after parties choose to accept;  also note that this is a single bit of the raw key, which actually will require two rounds of the protocol to derive.  In particular, two copies of $\ket{\psi}_{ABE}$ will be used to produce $\rho_{ABEM}$ using the S-CAD protocol. 

Our main result is stated below, in Theorem \ref{them:main}, which bounds Eve's uncertainty on Alice's final raw key bit, after CAD is run.  Note that we also must bound her uncertainty based on all public discussion (in particular, the parity announcements).

\begin{theorem}\label{them:main}
Let $\ket{\psi}_{ABE}$ be the state produced by Eve on a single round, and let $\rho_{ABEM}$ be the resulting density operator modeling a single raw key bit after S-CAD runs and all parties accept, as discussed above.  Then, assuming collective attacks, if the observed $Z$ basis error rate is $Q^Z_\Delta$ for all $\Delta\in\{0,1\}^p$, and the observed $X$ basis error rate is $Q_X$, it holds that $H(A|EM)_\rho$ is lower bounded by:
\begin{equation}
   H(A|EM)_\rho \ge \min_{\{\nu_\Delta\}} \left(p_a -\sum_{(b,c)\in A_{\cad}} Q^Z_{b}Q^Z_{c} h\left(\tau_{b,c}\right)\right)
 \end{equation}
 where $A_\cad$ is defined in Equation \ref{acadd};
$p_a$ is the probability of accepting a two-round block, and is defined in Equation \ref{eq:thm3};
and finally $\tau_{b,c}$ is defined in Equation \ref{eq:tau}.  The above expression is minimized over all $\nu_\Delta$, where $\Delta \in \{0,1\}^p$, subject to the  constraints that $0\le\nu_\Delta \le Q_\Delta^Z$ and $\sum_{\Delta\in\{0,1\}^p}\nu_\Delta = Q_x$.
\end{theorem}

\begin{proof}


  Let $\ket{\psi}_{ABE}$ be the state created by Eve on a single round of the quantum communication stage of the protocol.  Using results from \cite{epping2017multi,dur1999separability,dur2000classification}, we may assume the state $tr_E\kb{\psi}_{ABE}$ is diagonal in the GHZ basis (Alice and the Bobs may apply a symmetrization step to enforce this).  Thus, it is to Eve's advantage that she has a purification of this diagonal state.  We therefore may assume, to Eve's advantage, that the state is of the form:
  \begin{equation}
    \ket{\psi}_{ABE} = \sum_{\g{x}{y}\in\G_p}\ket{g(x;y)}\ket{E_{x,y}},
  \end{equation}
  where $\ket{E_{xy}}$ are sub-normalized, orthogonal states.  Given the observed error rates $Q^Z_\Delta$ and $Q_X$, it holds that:
  \begin{align}
    Q^Z_\Delta &= \bk{E_{\Delta, 0}} + \bk{E_{\Delta,1}}, \wforall \Delta\in\{0,1\}^p\label{eq:qdelta-constraint}\\
    Q_X &= \sum_{\Delta\in\{0,1\}^p}\bk{E_{\Delta,1}}.
  \end{align}
  The first constraint (for all $\Delta$) is easy to see; the second follows from basic properties of GHZ states, as discussed in Section \ref{section:prelim}.

We first model the delayed measurement CAD protocol on two copies of this state (since our S-CAD protocol requires two rounds).  Consider $\ket{\psi}_{ABE}^{\otimes 2}$ which, after permuting subspaces, can be written in the form:
\begin{equation}
  \ket{\psi}_{ABE}^{\otimes 2} \cong \sum_{\g{x}{y},\g{z}{w}\in\G_p}\ket{g(x;y)}\ket{g(z;w)}\ket{E_{x,y,z,w}},
\end{equation}
where we define $\ket{E_{x,y,z,w}} = \ket{E_{x,y}}\otimes\ket{E_{z,w}}$.  Note that the GHZ state indexed by $x$ and $y$ represent the ``Left'' system while the other GHZ state, indexed by $w$ and $z$, is the ``Right'' system for the S-CAD protocol (see Section \ref{sec:protocol}).


We now apply Lemma \ref{lem:DelayedMeasure} to the above state, which models the delayed measurement version of S-CAD.  The resulting state (before any measurements are performed), is:
\begin{equation}
    \begin{split}
    \sum_{\g{x}{y},\g{z}{w}\in \G_p}\ket{(x\oplus z) \wedge \CAD}_{rej}
    \otimes&\frac{1}{\sqrt{2}}\sum_{m=0}^1(-1)^{wm}\ket{m}_M  \\
    \otimes& \ket{g(x, m, z\oplus m^p; y\oplus w)}_{AB}\\
    \otimes&\ket{E_{x,y,z,w}}
   \end{split}
\end{equation}


Now, parties will measure the $rej$ register; if at least one qubit is in the state one, then this round is rejected.  Since we are only interested in analyzing the entropy in an accepted state, we condition on all parties signaling to accept (i.e., $(x\oplus z) \wedge \CAD = 0\cdots 0$).
Conditioned on accepting, the post-measured state collapses to:
\begin{align}
    \frac{1}{p_a}\sum_{(x,z) \in A_\CAD} P\big(\frac{1}{\sqrt{2}}\sum_{m,y,w\in\{0,1\}}&(-1)^{w\cdot m}\ket{m}_M  \\
                                                                  &\otimes \ket{g(x, m, z\oplus m^p; y\oplus w)}_{AB}\notag\\
                                                                  &\otimes \ket{E_{x,y,z,w}} \big)\notag
\end{align}
where $p_a$ is a normalization term, namely the probability of accepting, and:
\begin{align}
  A_\CAD&=\{(x,z) \in \{0,1\}^p\times\{0,1\}^p \st (x \oplus z) \wedge \CAD = 0\cdots 0\}\notag\\
        &= \{(x,z)\st x_i = z_i \text{ whenever } \cad_i = 1\}\label{acadd}
\end{align}
is the set of bit-error conditions that lead to an acceptance.  From the above, and the fact that Eve's vectors are orthogonal, it can be shown, using Equation \ref{eq:qdelta-constraint}, that:
\begin{align}
  p_a &=\sum_{(x,z)\in A_\cad}Q^Z_x Q^Z_z.\label{eq:thm3}
\end{align}
The above expression for $p_a$ also makes intuitive sense, since Bob$_i$ (with $\cad_i=1$) will accept only if both Left and Right GHZ states result in an error (relative to Alice, so $x_i=z_i=1$) or no error ($x_i=z_i=0$).  Of course, for those Bobs who have $\cad_i = 0$, they will always accept, and so no constraints on their measurement outcomes are placed.


We now measure all systems in the $Z$ basis, and trace out the Bobs' systems, as we're only interested, currently, in computing a bound on $H(A|EM)$.  We also trace out Alice's Right (her second) qubit which is discarded by S-CAD. Equivalently, we may simply discard the Bobs' systems, now, and measure Alice in the $Z$ basis (discarding Alice's second qubit measurement).  This measurement, leads to a mixed state, we denote $\rho_{AEM}$, and is found to be $\rho_{AEM} = \frac{1}{2}\sum_{m\in\{0,1\}}\kb{m}_M\otimes\rho_{AE}\up{m}$, where $\rho_{AE}\up{m}$ is the quantum state conditioned on message $m$ being sent by Alice.  This is readily computed as:
\begin{align*}
  &\rho_{AE}\up{m}=\frac{1}{2p_a} \sum_{a=0}^1[a]_A \otimes\\
  &\sum_{(x,z)\in A_\CAD} P\left(\sum_{y,w\in\{0,1\}}  (-1)^{(y\oplus w) \cdot a} (-1)^{w\cdot m}\ket{E_{x,y,z,w}}_E \right)
  \label{eq:rhoAE}
\end{align*}

Now, let's focus on a particular message $m$ and compute $H(A|E,M=m)_\rho = H(A|E)_{\rho\up{m}}$.  Equation \ref{eq:prelim:mixed} will allow us to then use this to bound the total entropy.  Since $m$ is public knowledge, and in particular, known to Eve, she (Eve) may apply the following unitary operator, $U_m$, to her ancilla state:
\begin{equation}
U_m \ket{E_{x,y,z,w}} = (-1)^{w\cdot m}\ket{E_{x,y,z,w}}.
\end{equation}
That this is unitary is trivial to show, since $\ket{E_{x,y,z,w}}$ are orthogonal states.  Furthermore, since unitary changes in basis do not affect entropy, we may apply this operator to Eve's system without affecting her uncertainty.  Note that, after applying this operator, we have $U_m\rho_{AE}\up{m}U_m^* = \rho_{AE}\up{0}$.  In this sense, while the message may leak additional information to Eve, each possible message will leak the same amount of information (which, intuitively, makes sense).  Thus, we actually analyze the entropy in the case $m=0$, since $H(A|E,M=0) = H(A|E,M=1)$.

Now, since $\ket{E_{x,y}}$ is a sub-normalized state, we may write $\ket{E_{x,y}} = \sqrt{\lambda_x^y}\ket{f_{x}^y}$, where $\lambda_x^y$ are non-negative real numbers (any other phase may be absorbed into the $\ket{f_x^y}$ vector).  Of course, we have $\bk{E_{x,y}} = \lambda_x^y$ (and, thus, the constraints from Equation \ref{eq:qdelta-constraint} can be translated to these $\lambda_x^y$ values).  We also have $\ket{E_{x,y,z,w}} = \sqrt{\lambda_{x}^{y}\lambda_z^w} \ket{f_{x}^{y}f_z^w}$.

Define the state:
\begin{align}
  \ket{F_{x,z}^{\up{a}}} &= \sum_{y,w\in\{0,1\}}(-1)^{(y\oplus w) \cdot a} \ket{E_{x,y,z,w}}_E\notag\\
                        &= \sum_{y,w\in\{0,1\}}(-1)^{(y\oplus w) \cdot a}\sqrt{\lambda_x^y\lambda_z^w}\ket{f_x^yf_z^w}.\label{eq:new-F}
\end{align}
Then we may write $\rho_{AE}\up{0} = \frac{1}{2p_a}\sum_a\kb{a}\sum_{(x,z)\in A_\cad}\kb{F_{x,z}\up{a}}$, thus allowing us to use Theorem \ref{thm:prelim}.  Of course, for this theorem to be applicable, we also require $\bk{F_{x,z}\up{a}}$ for all $(x,z)\in A_\cad$ and $a\in\{0,1\}$.  We also require $\braket{F_{x,z}\up{0}|F_{x,z}\up{1}}$.  Since Eve's ancilla vectors are orthogonal, using Equation \ref{eq:new-F}, these are readily computed:
\begin{align}\label{eq:fxza}
  \braket{F_{x,z}^{(a)}|F_{x,z}^{(a)}} &= \sum_{y,w\in\{0,1\}}\lambda_x^y\lambda_z^w = (\lambda_x^0+\lambda_x^1)(\lambda_z^0+\lambda_z^1) \\
                                       &= Q^Z_x Q^Z_z, \label{eq:innerprod1}
\end{align}
where, for the last equality, we used Equation \ref{eq:qdelta-constraint}.  Finally, we also have:
\begin{align}
\braket{F_{x,z}\up{0}|F_{x,z}\up{1}} = \sum_{y,w\in\{0,1\}} (-1)^{y\oplus w}\lambda_x^y\lambda_x^z = (\lambda_x^0-\lambda_x^1)(\lambda_z^0-\lambda_z^1)\label{eq:innerprod2}
\end{align}

Now, since $\lambda_\Delta^0 + \lambda_\Delta^1 = Q^Z_\Delta$ for all $\Delta\in\{0,1\}^p$ (see Equation \ref{eq:qdelta-constraint}), let $\nu_\Delta = \lambda_\Delta^1$ and, then, $\lambda_\Delta^0 = Q^Z_\Delta - \nu_\Delta$.  With this notation, Equation \ref{eq:innerprod2} simplifies to:
$\braket{F_{x,z}\up{0}|F_{x,z}\up{1}} = (Q^Z_x-2\nu_x)(Q^Z_z-2\nu_z).$
This gives us enough information to utilize Theorem \ref{thm:prelim} to bound $H(A|E,M=0)$ (which, as discussed above, is the same as $H(A|E,M=1)$). 
As shown in Equation \ref{eq:fxza}, $\braket{F_{x,z}^{(0)}|F_{x,z}^{(0)}} = \braket{F_{x,z}^{(1)}|F_{x,z}^{(1)}}$; note also that $p_a = \frac{1}{2}\sum_{x,z,a}\bk{F_{x,z}\up{a}}$. So after some simplification: 
\begin{equation*}
  H(A|E,M=0) \geq \min_{\{\nu_\Delta\}} \left(1 - \frac{1}{p_a}\sum_{(x,z)\in A_\CAD} Q^Z_{x} Q^Z_{z} h(\tau_{x,z})\right)
\end{equation*}
where:
\begin{equation}\label{eq:tau}
  \tau_{x,z} = \frac{1}{2}\left(1 - \frac{|(Q^Z_x-2\nu_{x})(Q^Z_{z}-2\nu_{z})|}{Q^Z_{x} Q^Z_{z}}\right)
\end{equation}
Note that we minimize over all $\nu_\Delta \in [0, Q^Z_\Delta]$, subject to the constraint that $\sum_\Delta\nu_\Delta = Q_X$ (both constraints follow from Equation \ref{eq:qdelta-constraint}), since we must assume Eve's original attack produces a state that minimizes her uncertainty.  This, combined with Equation \ref{eq:prelim:mixed}, completes the proof.
\end{proof}

\subsection{Final Key-Rate Derivation}\label{sec:asymproof:final-keyrate}

Theorem \ref{them:main} allows us to bound the conditional entropy in the state conditioned on all parties accepting.  To determine the final secret key rate, we will also need to determine the maximal error correction leakage.  Let $Q_{AB_j}$ be the bit error rate between Alice and Bob$_j$ 
before CAD is run, and let $Q_{AB_j}^{CAD}$ be the bit error rate between Alice and Bob$_j$ after the S-CAD protocol runs.
It is easy to see that:
$Q_{AB_j} = \sum_{\substack{\Delta\in\{0,1\}^p\\\stt \Delta_j=1}} Q_\Delta^Z.$
As shown in \cite{epping2017multi}, the error correction leakage term will be:
\begin{equation}\label{eq:asymproof:leakec}
  \leakEC = \max_jH(A|B_j) = \max_j h(Q_{AB_j}^{CAD}).
\end{equation}
where the maximum is over all $j=1,2\cdots, p$ (over all Bobs).

Returning to the actual key-rate expression, let $N$ be the total number of signals sent, and let $p_a$ be the probability of accepting any particular block of two (where $p_a$ is defined in Equation \ref{eq:thm3}).  Then, the total expected size of the raw key $n$, after S-CAD runs, will be $n = \frac{1}{2}p_aN$.  By Equation \ref{eq:prelim:keyrate}, we have the final key-rate in the asymptotic setting, then, is:
\begin{equation}
  \lim_{N\rightarrow\infty}\frac{\ell}{N} = \frac{p_a}{2}\left(H(A|EM) - \max_jh(Q_{AB_j}^{CAD})\right),
\end{equation}
where $H(A|EM)$ can be bounded using our Theorem \ref{them:main}.

Now, in general $Q_{AB_j}^{CAD}$ is an observable quantity that parties can estimate through standard sampling methods.  Since we are in the asymptotic setting, this sampling will not hurt the users in terms of efficiency.  For our evaluations (\S\ref{sec:results}), we will compute expected values for these based on the assumption that links in the star network generate noise independently. This is done just to simplify our evaluations. In practice, one would simply observe the actual values and use them in Equation \ref{eq:asymproof:leakec}.  If Bob$_j$ runs CAD (i.e., $\cad_j=1$), then it is easy to see that:
\begin{equation}
  Q_{AB_j}^{CAD} = \frac{\left(Q_{AB_j}\right)^2}{p_a}.
\end{equation}
On the other hand, if $\cad_j=0$, then since the links are assumed to be independent, it is expected that $Q_{AB_j}^{CAD} = Q_{AB_j}$ (since the probability of acceptance is not dependent on this link).  

Finally, we note that our result above was derived for collective attacks.  However, standard post-selection techniques \cite{christandl2009postselection} can be used to promote our analysis to general attacks.  One may also take advantage of the asymptotic equipartition property \cite{tomamichel2009fully}, along with post-selection techniques, to promote our analysis to the finite key setting.  However, there, one must be careful to work out the appropriate sampling errors.  We leave this as future work.


%% file: technical.tex
\subsection{Delayed Measurement S-CAD}\label{sec:technical:delayed-cad}


Recall that the delayed measurement CAD protocol operates after Alice and Bob sample their systems, but before measuring the remaining (unsampled) qubits in the $Z$ basis. The quantum state, in such an instance, may be represented as a mixture of pure states, each a superposition of GHZ states that can be written in the form:
\begin{align}\label{eq:technical:ghz-before}
  \sum_{\g{x}{y},\g{z}{w}\in J}&\ket{g(x_1;y_1)}\ket{g(z_1;w_1)}\cdots \ket{g(x_n;y_n)}\ket{g(z_n;w_n)}\\
  &\otimes\ket{E_{x,y,z,w}},\notag
\end{align}
where $J \subset\G_p^n\times \G_p^n$ (recall, this is the GHZ alphabet defined in Section \ref{section:prelim}).  We use $\g{x}{y}$ to index the Left GHZ states and $\g{z}{w}$ to index the Right GHZ states (recall, S-CAD divides the signal randomly into Left and Right portions).

The following Lemma analyzes the transformation of Equation \ref{eq:technical:ghz-before}, after running the delayed measurement protocol, but before any measurements are made.
\begin{lemma}
  Let $\ket{\psi}$ be a state of the form shown in Equation \ref{eq:technical:ghz-before}.  Then, after running the delayed measurement S-CAD protocol, but before measuring any system, the state evolves to:
  \begin{align}
  \begin{split}
    &\sum_{\g{x}{y},\g{z}{w}\in J}\ket{(x_1\oplus z_1) \wedge \CAD}_{rej_1}\cdots \ket{(x_n\oplus z_n) \wedge \CAD}_{rej_n}\\
    \otimes&\frac{1}{\sqrt{2^n}}\sum_{m\in\{0,1\}^n}(-1)^{w\cdot m}\ket{m}_M  \\
    \otimes& \ket{g(x_1, m_1, z_1\oplus m_1^p; y_1\oplus w_1)}_{AB}\\
    \vdots&\\
    \otimes& \ket{g(x_n, m_n, z_n\oplus m_n^p; y_n\oplus w_n)}_{AB}\\
    \otimes&\ket{E_{x,y,z,w}},
   \end{split}
  \end{align}
  where, above, $a\wedge b$ is the bit-wise AND operation; $m_i^p$ is a $p$-bit string consisting of the bit $m_i$ repeated $p$ times; and $w\cdot m$ is the bit-wise modulo two dot product.  Furthermore, by $x_i,m_i,z_i\oplus m_i^p$, we mean the bit-string concatenation of $x_i$ with $m_i$ with $z_i\oplus m_i^p$ (thus creating a $2p+1$ bit string).  Note that each individual GHZ state in the above consists of $2(p+1)$ qubits (representing Alice's and each Bobs' qubits for a particular CAD block which is two rounds of the protocol).
  \label{lem:DelayedMeasure}
\end{lemma}
\begin{proof}
  We first consider a single two-round block of the form $\ket{g(x_i;y_i)}\ket{g(z_i;w_i)}$ and trace the execution of the delayed measurement CAD on this.  Recall that $x_i,z_i\in\{0,1\}^p$ and $y_i,w_i\in\{0,1\}$.  The result will then follow by linearity.  First, an additional, empty, ancilla $\ket{0}_{M_i}\ket{0}_{rej_i}$ is added to the state, where the $M_i$ register consists of one qubit, and the $rej_i$ register consists of $p$ qubits.  Expanding the GHZ states in terms of the $Z$ basis and applying Alice's DCNOT operation (to model her parity computation), results in:
  \begin{align*}
    &\ket{0}_{M_i}\ket{0}_{rej_i}\ket{g(x_i;y_i)}\ket{g(z_i;w_i)} \\
    =&\ket{0}_{M_i}\ket{0}_{rej_i}\left(\frac{1}{2}\ket{0x_i0z_0} + (-1)^{y_i\oplus w_i}\ket{1\bar{x}_i 1\bar{z}_i}\right)\\
    +& \ket{0}_{M_i}\ket{0}_{rej_i}\left((-1)^{w_i}\frac{1}{2}\ket{0x_i1\bar{z}_0} + (-1)^{y_i}\ket{1\bar{x}_i 0z_i}\right)\\
    \mapsto &\ket{0}_{M_i}\ket{0}_{rej_i}\left(\frac{1}{2}\ket{0x_i0z_0} + (-1)^{y_i\oplus w_i}\ket{1\bar{x}_i 1\bar{z}_i}\right)\\
                                                                  +& (-1)^{w_i}\ket{1}_{M_i}\ket{0}_{rej_i}\left(\frac{1}{2}\ket{0x_i1\bar{z}_0} + (-1)^{y_i\oplus w_i}\ket{1\bar{x}_i 0z_i}\right)
  \end{align*}
  Each Bob $j$, will now check the parity of his system, across the ``$x$'' and ``$z$'' portion, to ensure it matches the value in $M_i$, assuming $\cad_j = 1$.  It is not difficult to see that this can also be modeled as each Bob, with $\cad_j = 1$, first applying his own DCNOT operation, and targeting the $j^{th}$ qubit in the $rej_i$ register.  Next, each Bob will apply a standard CNOT controlled on the $M_i$ register, and targeting the $j^{th}$ qubit in $rej_i$.  If the resulting value is a $\ket{0}$ (later, after measurement, of course), Bob$_j$'s signal is to accept this round; otherwise he signals to reject this round.  Of course, the above only occurs for those Bob's where $\cad_j = 1$; all other Bob's leave the $rej_j$ register to a $\ket{0}$ state (it's initial default value).  It is not difficult to see that this process will result in each Bob accepting this round, only if $x_i = z_i$, for all bits of these strings where $\cad_j = 1$.  Thus, the resulting system will be:
  \begin{align*}%
    &\ket{(x_i\oplus z_i) \wedge \cad}_{rej_i}\left(\ket{0}_{M_i}\left(\frac{1}{2}\ket{0x_i0z_0} + (-1)^{y_i\oplus w_i}\ket{1\bar{x}_i 1\bar{z}_i}\right)\right.\\
    & \left.+ (-1)^{w_i}\ket{1}_{M_i}\left(\frac{1}{2}\ket{0x_i1\bar{z}_0} + (-1)^{y_i\oplus w_i}\ket{1\bar{x}_i 0z_i}\right)\right)
  \end{align*}
Note that we permuted subspaces, above, to move the $rej$ register to the left.  Of course, this can be written in terms of GHZ states (now of $2(p+1)$ qubits):
  \begin{align*}
    \ket{(x_i\oplus z_i)\wedge \cad}_{rej_i}\sum_{m_i\in\{0,1\}^n}&\frac{1}{\sqrt{2}}(-1)^{m_i\cdot w_i}\ket{m_i}_{M_i}\\
    &\otimes\ket{g(x_i, m_i, z_i\oplus m_i^p; y_i\oplus w_i)}.
  \end{align*}
Since everything, above, was unitary, the result follows from linearity.
\end{proof}


%% file: eval.tex
\section{Evaluation Results} \label{sec:results}

We next present evaluation results in star topologies, where Alice is in the center of the star, distributing GHZ states to $p$ Bobs, which are the terminal nodes in the star. We vary $p$ to 2, 3, 4, and 7, and vary the channels from Alice to the Bobs to explore different settings. Henceforth, we refer to Bob${_i}$ as $B_i$ for simplicity.


For all our simulations, we specify a value for $Q_{AB_i}$, namely the noise between Alice and Bob$_i$ (see Section \ref{sec:asymproof:final-keyrate}).  To evaluate our key-rate, we assume each link acts independently in our simulation, which allows us to compute $Q_\Delta^Z$ as follows:
\begin{equation}
    Q_\Delta^Z = \prod_{i=1}^p (Q_{AB_i})^{\Delta_i}(1-Q_{AB_i})^{1-\Delta_i}
\end{equation}
Note that, in practice, Alice and the Bobs will sample all these values directly.

First, we compare 
our result when $\cad = 1\cdots 1$ (i.e., when CAD is enabled for all parties) to prior work \cite{krawec2025quantum}.  As seen in Figure \ref{fig:pw-comp}, our new proof of security asymptotically outperforms our earlier work.  Of course, our earlier work did not support the ``Selective'' option, and so we only compare it with $\cad = 1\cdots 1$.  
In this case, 
our new S-CAD protocol is identical to the one in \cite{krawec2025quantum}. The above result shows that our new asymptotic proof yields a tighter key-rate bound.

\begin{figure}
    \centering
    \includegraphics[width=0.48\linewidth]{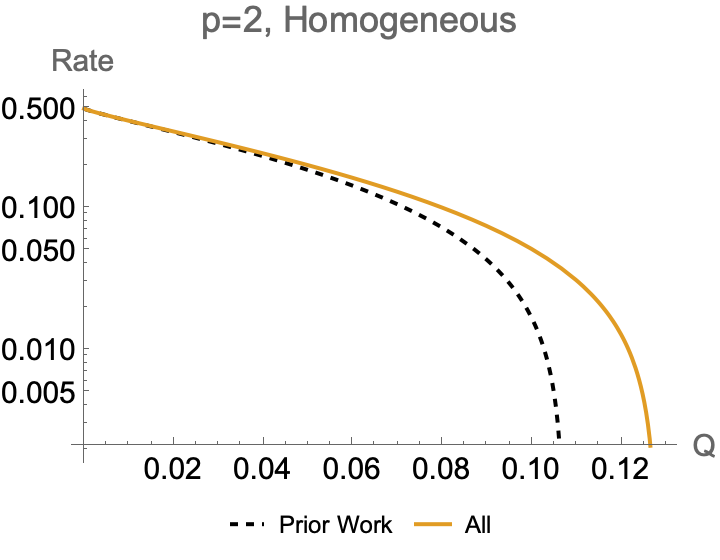}
    \includegraphics[width=0.48\linewidth]{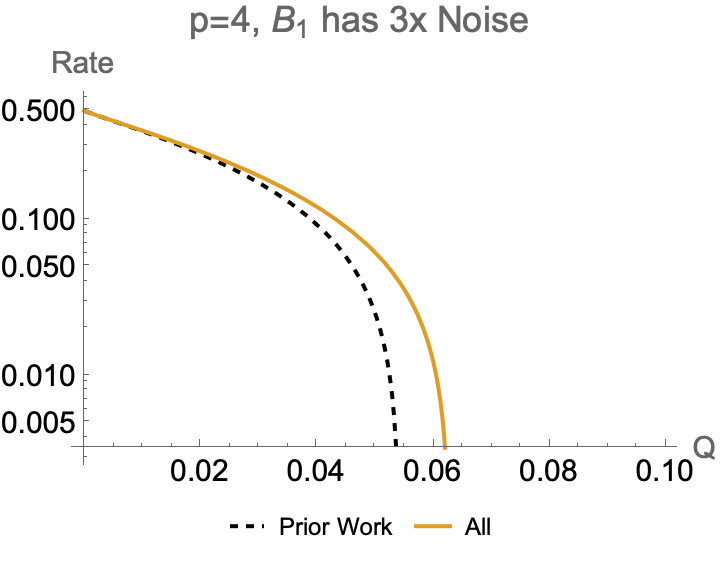}
    \caption{Comparing our work (Solid Line) to prior work in \cite{krawec2025quantum} (Dashed Line).  In all our tests, our new result outperformed prior work, when CAD is enabled for all parties.}
    \label{fig:pw-comp}
\end{figure}

In the following, we evaluate our S-CAD protocol in a variety of settings, and compare it to the original QCKA protocol of \cite{grasselli2018finite} that does not include CAD.
The asymptotic key-rate of this original QKCA protocol
is simply $1 - h(Q_X) - \max_jh(Q_{AB_j})$.

\heading{Results for $p=2$}
Figures~\ref{fig:2Bobs}a, b and c present the results with two Bobs (thus, three parties total), where the channel from Alice to $B_2$ has noise $Q$, while the noise from Alice to $B_1$ is set to $Q$, $2Q$, and $3Q$, respectively. We set $Q_X = Q$.  For each scenario, we compare two cases, when both $B_1$ and $B_2$ perform CAD (i.e., $\cad=11$), and when
only $B_1$ (who has channel noise equal or higher than that of $B_2$) performs CAD (i.e., $\cad = 10$).

We see that in the homogeneous settings, having both Bobs perform CAD is always more beneficial than having only one with CAD enabled.  This makes sense since, when only one Bob uses CAD, there is no advantage in the error correction leakage term (since the maximal noise will be dominated by the party with CAD disabled).  Of course, if the overall noise is low, then turning off CAD entirely is optimal.


In heterogeneous settings, when $B_1$ has twice the amount of noise as $B_2$, we still observe that having both Bobs perform CAD is more beneficial than selective CAD (i.e., only $B_1$ performs CAD). However, when $B_1$ has three times  amount of noise as $B_2$, we observe that selective CAD can outperform the case where all parties use CAD, but only slightly at the high noise scenario.  This makes sense, given the interplay between the probability of acceptance, $p_a$ (Equation \ref{eq:thm3}), and the overall error correction leakage.  Note that the fewer parties using CAD, the higher $p_a$ will be.


In summary, for the three party (two Bob) case, it is generally best to either enable CAD for all Bobs, or disable CAD entirely, depending on the overall noise level.  If one of the two Bobs has significantly higher noise than the other, it may be beneficial to enable CAD only for that one Bob and not the other.

\begin{figure} 
\center{
   \begin{subfigure}{0.8\linewidth}
       \includegraphics[width=\linewidth]{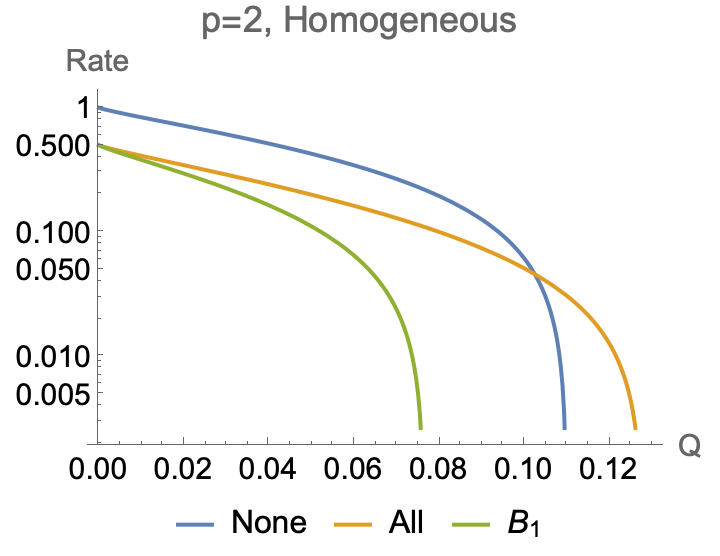}
       \caption{Homogeneous case.}
   \end{subfigure}}
\hfill 
   \begin{subfigure}{0.48\linewidth}
       \includegraphics[width=\linewidth]{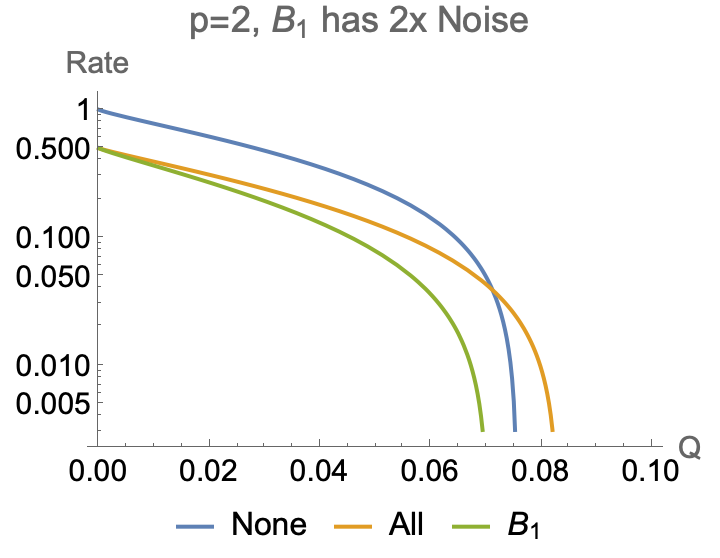}
       \caption{$B_1$ has 2x noise as $B_2$.}
   \end{subfigure}
\hfill 
   \begin{subfigure}{0.48\linewidth}
       \includegraphics[width=\linewidth]{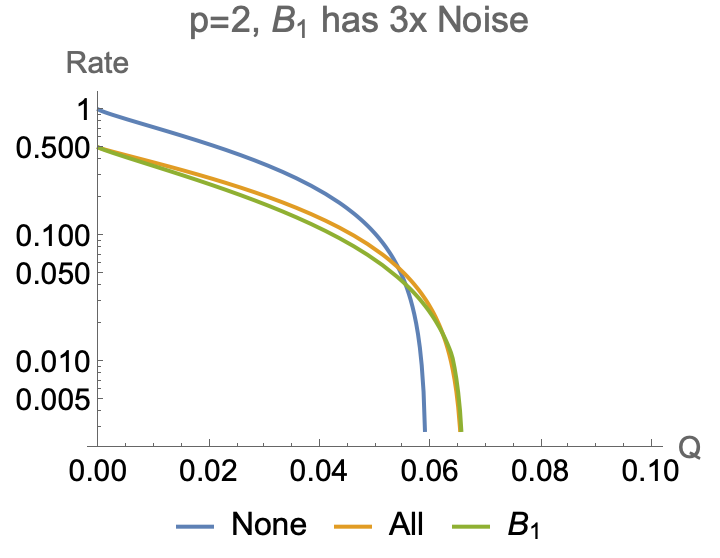}
       \caption{$B_1$ has 3x noise as $B_2$.}
   \end{subfigure}
   \caption{Three party scenario (Alice and two Bobs).  Here, and in other figures, ``None'' is the result of running the original QCKA protocol without CAD, using the key-rate expression from \cite{grasselli2018finite}.
   }
   \label{fig:2Bobs}
\end{figure}


\begin{figure*}[ht] 
   \begin{subfigure}{0.32\textwidth}
       \includegraphics[width=\linewidth]{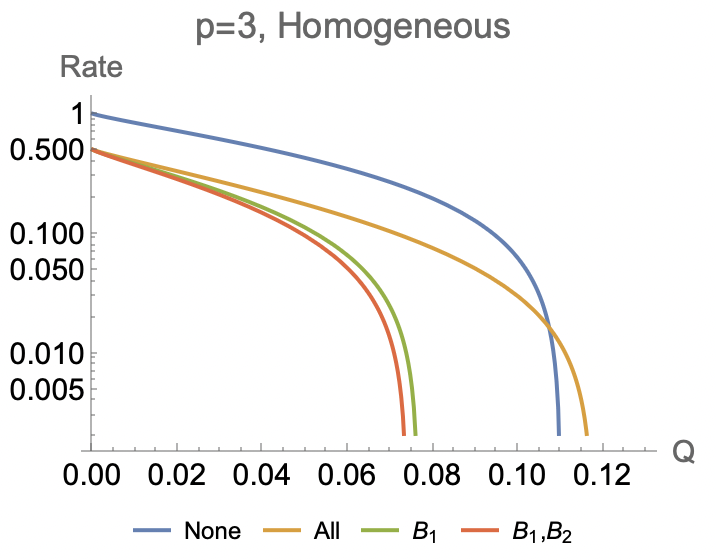}
       \caption{Homogeneous case.}
   \end{subfigure}
\hfill 
   \begin{subfigure}{0.32\textwidth}
       \includegraphics[width=\linewidth]{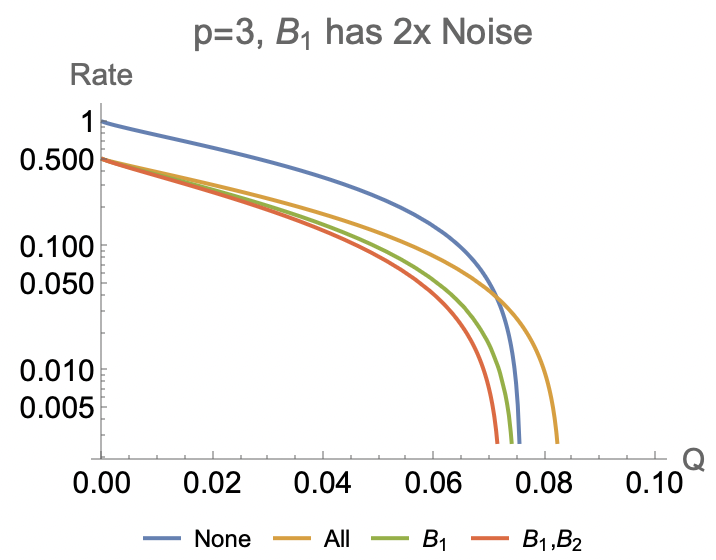}
       \caption{1 bad channel: $B_1$ has noise $2Q$, and all the other Bobs have noise $Q$.}
   \end{subfigure}
\hfill 
   \begin{subfigure}{0.32\textwidth}
       \includegraphics[width=\linewidth]{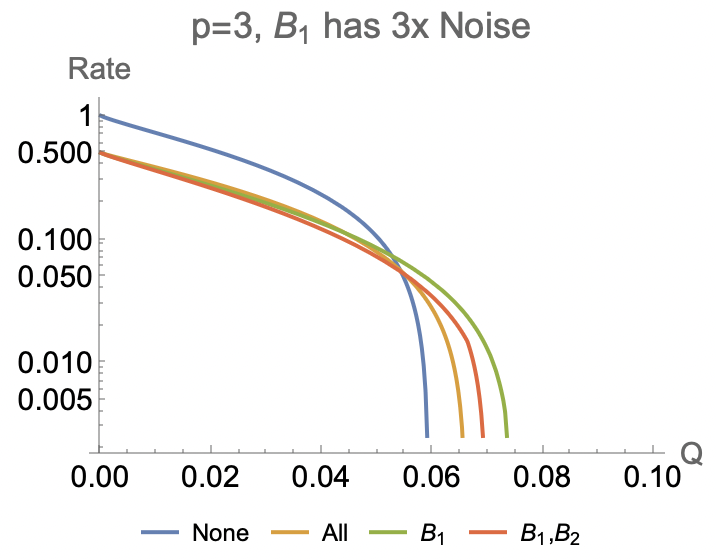}
       \caption{1 bad channel: $B_1$ has noise $3Q$, and all the other Bobs have noise $Q$.}
   \end{subfigure}
      \begin{subfigure}{0.32\textwidth}
       \includegraphics[width=\linewidth]{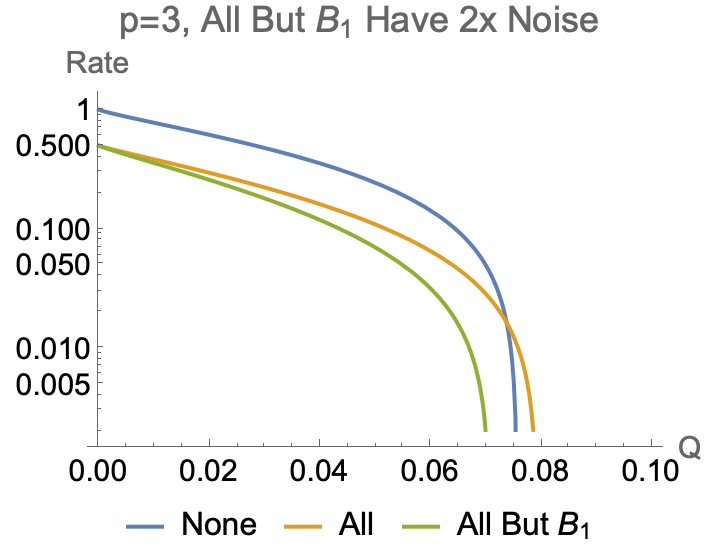}
       \caption{1 good channel: $B_1$ has noise $Q$, and all the other Bobs have noise $2Q$. 
       }
   \end{subfigure}
      \begin{subfigure}{0.32\textwidth}
       \includegraphics[width=\linewidth]{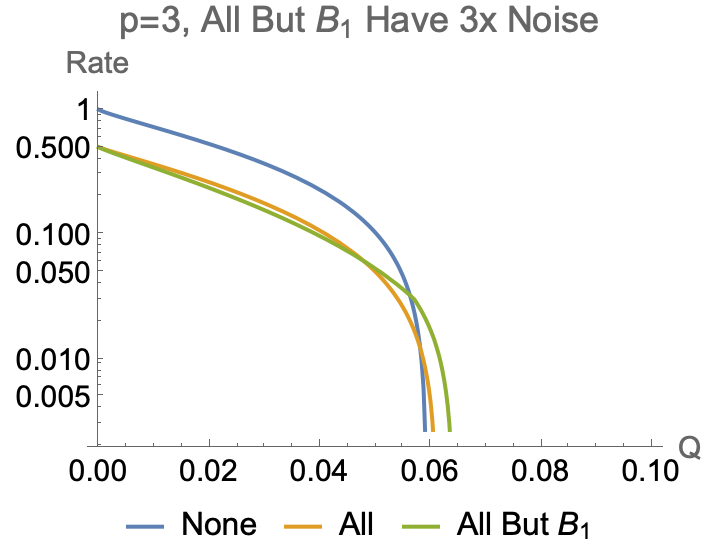}
       \caption{1 good channel: $B_1$ has noise $Q$, and all the other Bobs have noise $3Q$.}
   \end{subfigure}
      \begin{subfigure}{0.32\textwidth}
       \includegraphics[width=\linewidth]{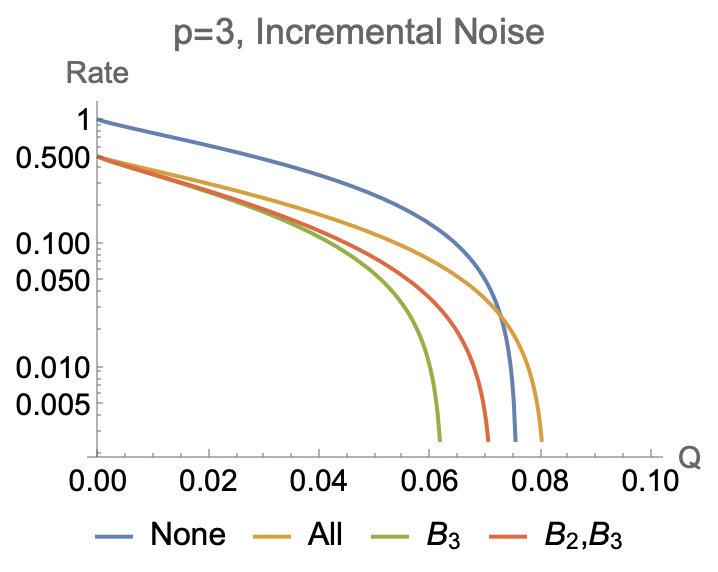}
       \caption{$B_1$, $B_2$ and $B_3$ have noises $Q$, $1.5Q$, and $2Q$, respectively.}
   \end{subfigure}
   \caption{Four party scenario (Alice and three Bobs).
   }
   \label{fig:3Bobs}
\end{figure*}
\begin{figure*}[ht] 
   \begin{subfigure}{0.32\textwidth}
       \includegraphics[width=\linewidth]{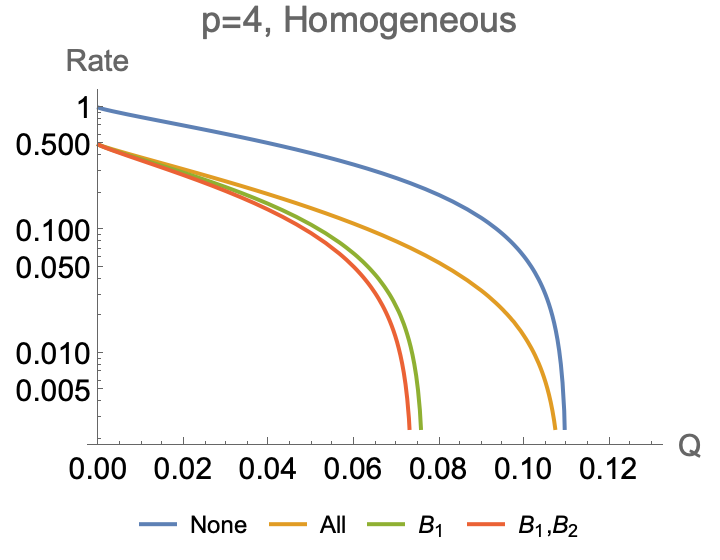}
       \caption{Homogeneous case.}
   \end{subfigure}
\hfill 
   \begin{subfigure}{0.32\textwidth}
       \includegraphics[width=\linewidth]{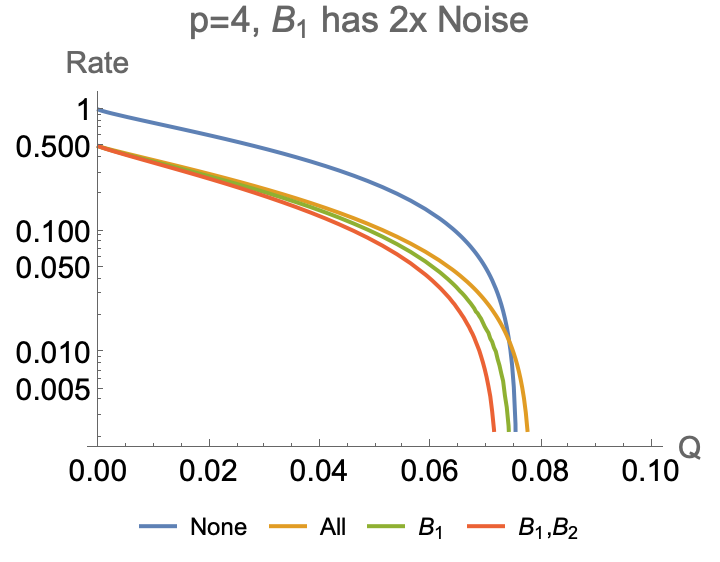}
       \caption{1 bad channel: $B_1$ has noise $2Q$, and all the other Bobs have noise $Q$.}
   \end{subfigure}
\hfill 
   \begin{subfigure}{0.32\textwidth}
       \includegraphics[width=\linewidth]{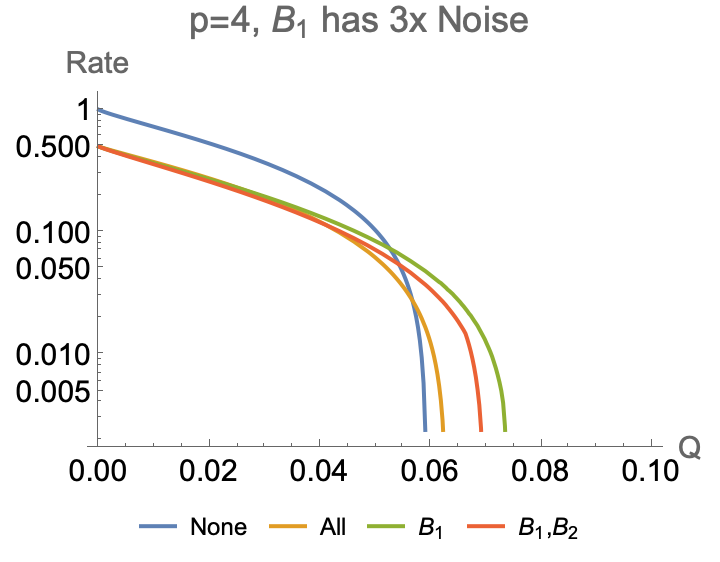}
       \caption{1 bad channel: $B_1$ has noise $3Q$, and all the other Bobs have noise $Q$.}
   \end{subfigure}\\
   %
   \caption{Five party scenario (Alice and four Bobs).
   }
   \label{fig:4Bobs}
\end{figure*}

\heading{Results for $p=3$}
We now present the results when $p=3$ (three Bobs, thus four parties total). Specifically, we consider the following four cases: (i) homogeneous channel, i.e., the channels from Alice to all Bobs have the same noise $Q$, (ii) single bad channel: the channels from Alice to both $B_2$ and $B_3$ have the same noise, $Q$, while the the channel from Alice to $B_1$ has noise set to $2Q$ or $3Q$, (iii) single good channel: the channel from Alice to $B_1$ has noise $Q$, while the the channels from Alice to both $B_2$ and $B_3$ have the same noise, set to $2Q$ or $3Q$, and (iv) incremental noise  setting: the channels from Alice to the three Bobs are set to $Q$, $1.5Q$, and $2Q$ respectively. Again, we set $Q_X = Q$ for each case.  In all the above cases, we consider two strategies: (1) all Bobs perform CAD; and (2) a subset of Bobs have CAD enabled (typically the higher noise Bobs).

We again observe that in the homogeneous setting (Fig.~\ref{fig:3Bobs}a), having CAD enabled for all Bobs produces an optimal result when the channel noise is high.  Note, however, that the advantage compared to not using CAD at all, is not as large
as it is in the $p=2$ case (Fig.~\ref{fig:2Bobs}a).  This is due to the fact that, as the number of parties increase, the probability of acceptance, $p_a$ Equation \ref{eq:thm3}, necessarily decreases.  Thus, in the homogeneous setting, as the number of parties increases, the advantage to using S-CAD in the multi-party setting diminishes.  {We also see that, in the homogeneous case, it is best to either enable CAD for all parties, or disable CAD entirely, depending on the channel noise.  As seen in Fig.~\ref{fig:3Bobs}a, turning on CAD only for $B_1$ severely hampers the protocol's performance.  This, however, makes sense since turning CAD on for only one of the parties will not give any advantage in error correction leakage for  the homogeneous setting (the party with CAD turned off will dominate this term); yet there is a disadvantage in that the probability of accepting $p_a$ will decrease.  When turning on CAD for a subset of parties (see the $B_1B_2$ curve), the argument is similar, however the $p_a$ term decreases even more, thus the $B_1B_2$ curve drops below the $B_1$ curve.  When turning CAD on for all parties, $p_a$ drops more, however now there is an advantage in error correction leakage and, thus, when the noise is high enough, the protocol can outperform the standard ``No CAD'' scenario.}



In the single bad channel case (Figures \ref{fig:3Bobs}b and c), we see that when the noise difference increases (e.g., the noisy link is $3Q$), it becomes optimal to only enable CAD for the one noisy Bob; when the difference is not as extreme (e.g., $2Q$), then it is advantageous to enable CAD across all parties, or to not use CAD at all, depending on $Q$.  
The above results are again due to the impact of error correction leakage and $p_a$.
In Fig.~\ref{fig:3Bobs}c, when the noise is high enough for one party, it makes sense to turn CAD on only for that party, as the error correction leakage will decrease, which will outweigh the disadvantage of a lower accepting probability, $p_a$.
The more parties with CAD enabled, the lower $p_a$ will be, and hence we observe  worse result under $B_1B_2$ (i.e., both $B_1$ and $B_2$ turn on CAD) compared to only $B_1$ uses CAD. On the other hand, compared to the strategy in which all parties turn on CAD, $B_1B_2$ leads to better results due to higher $p_a$ and similar error correction leakage.   

A similar trend is shown in Figures \ref{fig:3Bobs}d through e, where, now, there is one ``good'' channel, and the rest have high noise.  Namely, when the difference in noise is high enough (e.g., $3Q$ and not $2Q$), it is advantageous to disable CAD for the low-noise party, but enable it for all high-noise parties. In other words, selective CAD is more beneficial than all the parties performing CAD when the channels among the parties differ more significantly (see Fig.~\ref{fig:3Bobs}e when the two Bobs have 3x noise than $B_1$), while the opposite is true when their noise levels are not substantially different (see Fig.~\ref{fig:3Bobs}d).


Similar results are observed in Fig.~\ref{fig:3Bobs}f, when $B_1$, $B_2$ and $B_3$ have noises $Q$, $1.5Q$, and $2Q$, respectively.  Since the channels of the three parties do not differ substantially, we see that it is better to enable CAD for all parties, when the noise increases.

\heading{Results for $p=4$}
Similar trends are seen for $p=4$ (four Bobs, thus five parties total), namely that when the noise is significantly higher for some Bobs and not others, it is best to enable CAD only for those noisy Bobs; see Fig. \ref{fig:4Bobs}.  Note that, in the homogeneous case, S-CAD always hurts performance, showing, again, that as the number of parties increases, the advantages to S-CAD 
diminishes, especially in the homogeneous setting.  When there are links with high noise, then S-CAD can still outperform no-CAD, by carefully enabling and disabling CAD for some parties.  

It is an open problem to design a multi-user CAD protocol (or prove that one cannot exist) that can continue to outperform the ``No CAD'' case, for high noise, homogeneous, or nearly homogeneous, networks, as the number of parties greatly increases.  In our protocol, $p_a$ begins to diminish rapidly as the number of parties running CAD increases.  However, in the 
heterogeneous case, with all but a few parties enabling CAD, the diminishing effect of $p_a$, is  outweighed by the increasing advantage in lower error correction leakage from S-CAD.

\heading{Results for $p=7$}
We see the same trends exhibit themselves for a larger number of parties, when $p=7$, in Fig.~\ref{fig:7Bobs}.  Specifically, in the homogeneous case, it is best to disable CAD entirely.  However, when there are some parties who have a significantly higher error rate, it is advantageous to enable CAD for that party, and disable it for the others.

\begin{figure} 
\center{
   \begin{subfigure}{0.48\linewidth}
       \includegraphics[width=\linewidth]{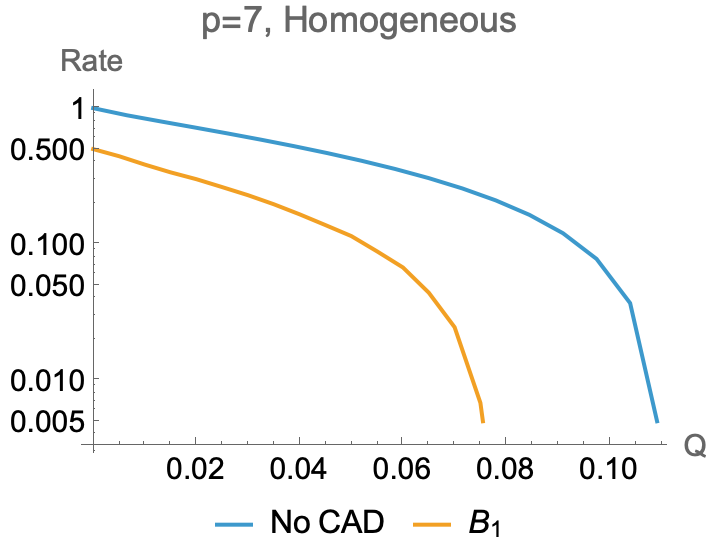}
       \caption{Homogeneous case.}
   \end{subfigure}}
\hfill 
   \begin{subfigure}{0.48\linewidth}
       \includegraphics[width=\linewidth]{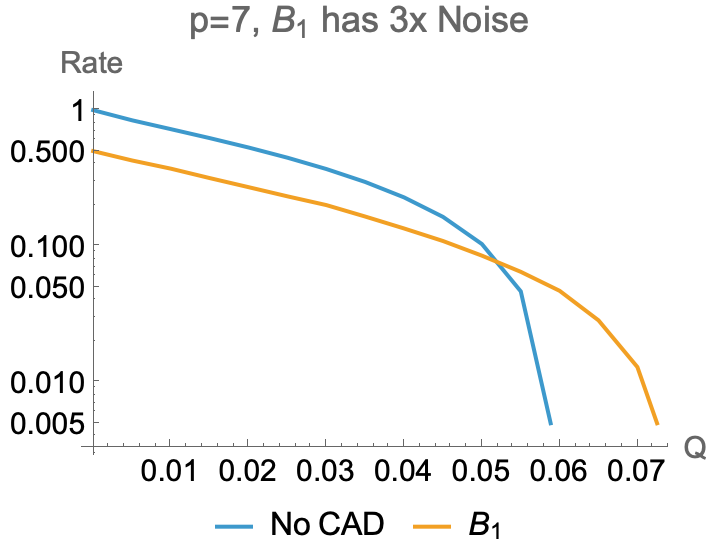}
       \caption{1 bad channel with $3Q$ noise}
   \end{subfigure}
\hfill 
   \caption{Results for $p=7$ (eight parties total).
   }
   \label{fig:7Bobs}
\end{figure}

%% file: closing.tex
\section{Closing Remarks}

In this paper, we introduced S-CAD, a novel CAD protocol  specifically designed for multi-user QCKA protocols.  We derived an information theoretic proof of security for this protocol against general attacks, and computed its asymptotic key-rate.  Finally, we evaluated our work in a variety of settings, discovering important lessons on when CAD can be beneficial, and when it should be disabled.

Many interesting future problems remain open.  Designing a more efficient S-CAD protocol would be interesting.  For instance, for any Bob who has CAD disabled, can his Right qubit be used for some other purpose (maybe random number generation, or a ``sub-group'' key)?  It would also be interesting to design a protocol that does not require information on $Q^Z_\Delta$ for all $\Delta\in\{0,1\}^p$, as this could improve practical performance.


\heading{Acknowledgments}
WOK and TT would like to acknowledge support from the NSF under grant number 2143644.
